\documentclass[11pt,a4paper]{article}
\usepackage{fullpage}
\usepackage{microtype}
\usepackage{amssymb,amsmath,amsthm} \usepackage{graphicx}
\usepackage{color} \usepackage{xspace}
\usepackage[numbers]{natbib} \usepackage{hyperref,color}
\usepackage{xspace} 
\usepackage[T1]{fontenc}
\graphicspath{{figures/}}
\newtheorem{theorem}{Theorem} 
\newtheorem{lemma}[theorem]{Lemma}

\newtheorem{corollary}[theorem]{Corollary}
\newtheorem{observation}[theorem]{Observation}

\title{Computing the Center Region and Its Variants\footnote{
This research was supported by the MSIT(Ministry of Science and ICT), Korea, under the SW Starlab support program(IITP-2017-0-00905) supervised by the IITP(Institute for Information \& communications Technology Promotion.)}}

\author{Eunjin Oh\thanks{Pohang University of Science and Technology,
		Korea. Email: {\tt{\{jin9082, heekap\}@postech.ac.kr}}} \and
	 Hee-Kap Ahn\footnotemark[2]~\thanks{Corresponding author.} }

\newcommand{\intersection} {\textsc{Intersection}}
\newcommand{\ch}{\textsf{CH}}
\def\polylog{\operatorname{polylog}}
\begin{document}
\date{}
\maketitle
\begin{abstract}
  We present an $O(n^2\log^4 n)$-time algorithm for computing the
  center region of a set of $n$ points in the three-dimensional
  Euclidean space. This improves the previously best known algorithm
  by Agarwal, Sharir and Welzl, which takes $O(n^{2+\epsilon})$ time
  for any $\epsilon > 0$.  It is known that the combinatorial
  complexity of the center region is $\Omega(n^2)$ in the worst case,
  thus our algorithm is almost tight. We also consider the problem of
  computing a colored version of the center region in the
  two-dimensional Euclidean space and present an $O(n\log^4 n)$-time
  algorithm.
\end{abstract}

\section{Introduction}
Let $S$ be a set of $n$ points in $\mathbb{R}^d$.  The \emph{(Tukey)
depth} of a point $x$ in $\mathbb{R}^d$ with respect to $S$ is
defined to be the minimum number of points in $S$ contained in a
closed halfspace containing $x$.  A point in $\mathbb{R}^d$ of largest
depth is called a \emph{Tukey median}.  
The Helly's theorem implies that the depth of a Tukey median is at
least $\lceil n/(d+1) \rceil$. In other words, there always exists a
point in $\mathbb{R}^d$ of depth at least $\lceil n/(d+1) \rceil$.
Such a point is called a \emph{centerpoint} of $S$.
A Tukey median is a centerpoint, but not every centerpoint is a Tukey median.
We call the set of all centerpoints in $\mathbb{R}^d$ 
the \emph{center region} of $S$.
Figure~\ref{fig:intro}(a) shows 9 points in the plane and their Tukey medians and
center region.

\begin{figure}[t]
  \begin{center}
    \includegraphics[width=0.7\textwidth]{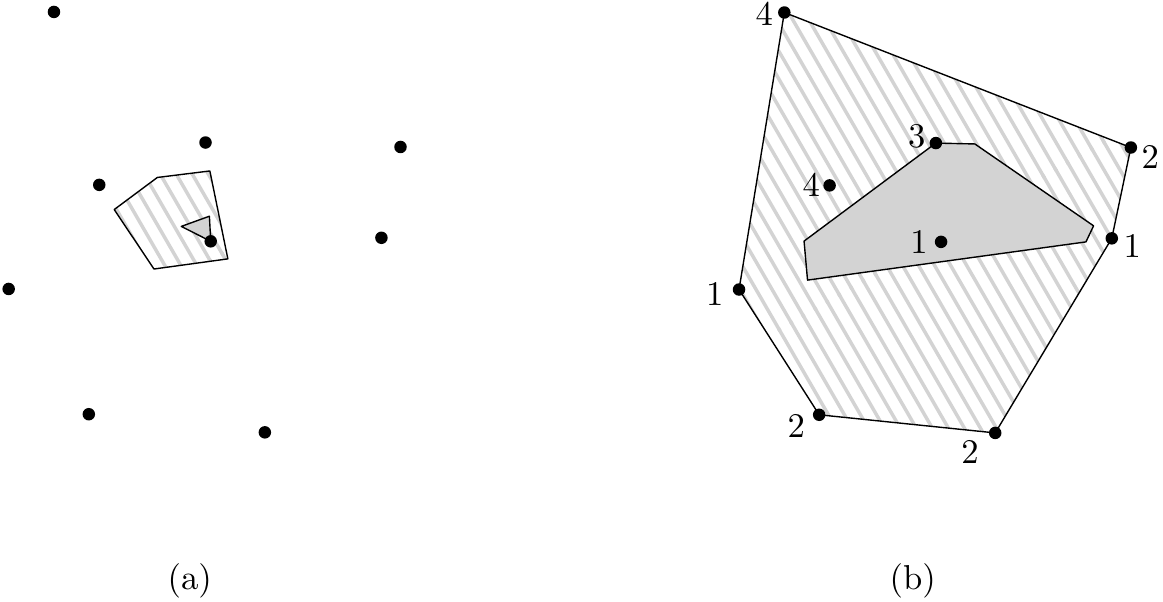}
    \caption{\small
      (a) 9 points in the plane ($n=9$ and $d=2$), and their Tukey medians
      (any point in the inner gray triangle) and center region (the outer pentagon).
      (b) 9 points each assinged one of 4 colors $\{1,2,3,4\}$ in the plane ($k=4$ and $d=2$),
      and their colorful Tukey medians (any point in the inner gray hexagon) and
      colorful center region (the outer hexagon). 
			\label{fig:intro}}
 \end{center}
\end{figure}

The Tukey median and centerpoint  are commonly used measures of the properties of
a data set in statistics and computational geometry. They are considered as
generalizations of the standard median in the one-dimensional space to a higher
dimensional space.
For instance, a good location for a hub with respect to a set of facilities
given in the plane or in a higher dimensional space would be the \textit{center} of
the facilities with respect to the distribution of the them in the underlying space.
Obviously, a Tukey median or a centerpoint of the facilities is a good candidate for
a hub location.
Like the standard median, the Tukey median and centerpoint are robust against outliers
and describe the center of data with respect to the data distribution.
Moreover, they are invariant under affine transformations~\cite{tukey-median}.

In this paper, we first consider the problem of computing the center region of points
in $\mathbb{R}^3$. By using a duality transform and finer triangulations in
the arrangement of planes, we present an algorithm which improves the best known
one for the problem.

Then we consider a variation of the center region where each point is given
one color among $k \in \mathbb{N}$ colors. Suppose that there are $k$
different types of facilities and we have $n$ facilities
of these types with $k\leq n$.  Then the standard definitions of
the center of points such as centerpoints, center regions, and
Tukey medians do not give a good representative of the $n$ facilities
of these types.  
Another motivation of colored points comes from discrete imprecise
data.  Suppose that we have imprecise points. We do not know the
position of an imprecise point exactly, but each imprecise point has
a candidate set of points where it lies. 
For such imprecise points, we can consider the points in the same
candidate set to be of the same color while any two points from two different candidate 
sets have different colors.

Then the \emph{colorful (Tukey) depth} of $x$ in $\mathbb{R}^d$ is naturally defined to be 
the minimum number of different colors of points contained in a closed 
halfspace containing $x$. A \emph{colorful Tukey median} is a point in $\mathbb{R}^d$
of largest colorful depth.
The colorful depth and the colorful Tukey median have properties
similar to the standard depth and Tukey median, respectively.  We prove that the
colorful depth of a colorful Tukey median is at least $\lceil k/(d+1)
\rceil$.  Then the colorful centerpoint and colorful center region are
defined naturally.  We call a point in $\mathbb{R}^d$ with colorful
depth at least $\lceil k/(d+1) \rceil$ a \emph{colorful centerpoint}.
The set of all colorful centerpoints is called the
\emph{colorful center region}.
Figure~\ref{fig:intro}(b) shows 9 points each assigned one of 4 colors
in the plane and their colorful Tukey medians and
colorful center region.

\paragraph{Previous work.}
In $\mathbb{R}^2$, the first nontrivial algorithm for computing a
Tukey median is given by Matous\v{e}k~\cite{center-region-2d}.  Their
algorithm computes the set of all points of Tukey depth at least a
given value as well as a Tukey median.  It takes $O(n\log^5
n)$ time for computing a Tukey median and $O(n\log^4 n)$ time for
computing the region of Tukey depth at least a given value.
	
Although it is the best known algorithm for computing the region of
Tukey depth at least a given value, a Tukey median can be computed
faster.  Langerman and Steiger~\cite{center-point-2d-langerman}
present an algorithm to compute a Tukey median of points in
$\mathbb{R}^2$ in $O(n\log^3 n)$ deterministic time.  Later,
Chan~\cite{tukey-median} gives an algorithm to compute a Tukey median
of points in $\mathbb{R}^d$ in $O(n \log n + n^{d-1})$ expected time.
	
A centerpoint of points in $\mathbb{R}^2$ can be computed in linear
time~\cite{center-point-2d}.  On the other hand, it is not known
whether a centerpoint of points in $\mathbb{R}^d$ for $d > 2$ can be
computed faster than a Tukey median.
	
The center region of points in $\mathbb{R}^2$ can be computed using
the algorithm by Matous\v{e}k~\cite{center-region-2d}.  For
$\mathbb{R}^3$, Agarwal, Sharir and Welzl
present an $O(n^{2+\epsilon})$-time algorithm for any $\epsilon>0$~\cite{center-region-3d}.
However, the constant hidden in the big-O notation is proportional to
$\epsilon$.  Moreover, as $\epsilon$ approaches $0$, the constant goes to
infinity.  It is not known whether the center region of points in
$\mathbb{R}^d$ for $d>3$ can be computed faster than an $O(n^d)$-time trivial algorithm 
which uses the arrangement of the dual hyperplanes of the points.
Moreover, even the tight combinatorial complexity of the center region is not known for $d>3$.
	
The center of colored points and
its variants have been studied in the literature~\cite{color-rectangle, color-disk, color-square}.
However, the centers of colored points defined in most previous
results are sensitive
to distances, which are not adequate to handle imprecise data.
Therefore, a more robust definition of a center of colored points is
required.  We believe that the colorful center region and colorful
Tukey median can be alternative definitions of the center of colored
points.
	
\paragraph{Our result.}
We present an algorithm to compute
the center region of $n$ points in $\mathbb{R}^3$ in
$O(n^2 \log^4 n)$ time. This improves the previously best known
algorithm~\cite{center-region-3d} and answers to the question posed
in the same paper. Moreover, it is almost tight
as the combinatorial complexity of the center region in $\mathbb{R}^3$ is $\Theta(n^2)$
in the worst case~\cite{center-region-3d}.

We also present an algorithm to
compute the colorful center region of $n$ points in $\mathbb{R}^2$ in 
$O(n\log^4 n)$ time.  We obtain this algorithm by
modifying the algorithm for computing the standard center region of
points in $\mathbb{R}^2$ in~\cite{center-point-2d}.

We would like to mention that a colorful Tukey median can
be computed by modifying the algorithms for the standard version of
a Tukey median without increasing the running times, which take $O(n\log n +
n^{d-1})$ expected time in $\mathbb{R}^d$~\cite{tukey-median} and
$O(n\log ^3 n)$ deterministic time in
$\mathbb{R}^2$~\cite{center-point-2d-langerman}.

\section{Preliminaries}
In this paper, we use a duality transform that maps a set of input
points in $\mathbb{R}^3$ to a set of planes.  Then we transform each problem into
an equivalent problem in the dual space
and solve the problem using the arrangement of the
planes.  The Tukey depth is closely related to the level of an
arrangement.  This is a standard way to deal with the Tukey
depth~\cite{center-region-3d,tukey-median,center-point-2d-langerman,center-region-2d}.  Thus, in this
section, we introduce a duality transform and some definitions for an
arrangement.

\paragraph{Duality transform.}
A standard duality transform maps a point $x \in \mathbb{R}^d$ to a 
hyperplane $x^*=\{z \in \mathbb{R}^d \mid \langle x, z \rangle = 1\}$ and
vice versa, where $\langle x, z \rangle$ is the scalar product of $x$
and $z$ for any two points $x, z \in \mathbb{R}^d$.  Then $x$ lies
below a hyperplane $s$ if and only if the point $s^*$ lies below the
hyperplane $x^*$.

\paragraph{Level of an arrangement.}
Let $H$ be a set of hyperplanes in $\mathbb{R}^d$.  A point $x \in
\mathbb{R}^d$ has \emph{level $i$} if exactly $i$ hyperplanes lie
below $x$ (or pass through $x$.) Note that any point in the same cell
in the arrangement of $H$ has the same level.  For an integer $\ell >
0$, the \emph{level} $\ell$ in the arrangement of $H$ is defined as the set of
all points of level at most $\ell$.  We define the level of an
arrangement of a set of $x$-monotone polygonal curves in $\mathbb{R}^2$ in a similar way.

\section{Computing the Center Region in
  \texorpdfstring{$\mathbb{R}^3$}{R3}}
\label{sec:first-problem}
Let $S$ be a set of $n$ points in $\mathbb{R}^3$.  In this section, we
present an $O(n^2\log^4 n)$-time algorithm for computing the set of
points of Tukey depth at least $\ell$ with respect to $S$ for a given
value $\ell$. We achieve our algorithm by modifying the previously best known algorithm 
for this problem given by Agarwal, Sharir and Welzl~\cite{center-region-3d}. Thus we first provide
a sketch of their algorithm.

\subsection{The Algorithm by Agarwal, Sharir and Welzl}
Using the standard duality transform, they map the
set $S$ of points to a set $S^*$ of planes in $\mathbb{R}^3$.
Due to the properties of the duality transform, the problem reduces to computing the convex hull of $\Lambda_\ell$,
where $\Lambda_\ell$ is the level $\ell$ in the arrangement of the planes in $S^*$.
The complexity of $\Lambda_\ell$ is $\Theta(n^2)$. Moreover,
the complexity of $\Lambda_\ell\cap h$ is $\Theta(n^2)$ for a plane $h$ in $S^*$ in the worst case.
Thus, instead of handling $\Lambda_\ell$ directly, they compute a convex polygon $K_h$ for each plane
$h \in S^*$ with the property that $\ch(\Lambda_\ell\cap h) \subseteq
K_h \subseteq \ch(\Lambda_\ell)\cap h$. Notice that $K_h$ is contained in $h$ for any $h\in S^*$. 
By definition, the convex hull
of $K_h$'s over all planes $h$ in $S^*$ is the convex hull of
$\Lambda_\ell$.  Therefore, once we have such a convex polygon $K_h$ for
every plane $h$, we can compute the set of points of Tukey depth at
least $\ell$ with respect to $S^*$.

To this end, they sort the planes in $S^*$ in the following order.
Let $h^+$ be the closed halfspace bounded from below by a plane $h$,
and $h^-$ be the closed halfspace bounded from above by $h$.
We use $\langle h_1,\ldots,h_n\rangle$ to denote the sequence of the
planes in $S^*$ sorted in the order satisfying the
following property: the level of a point $x \in
h_i$ in the arrangement of $S^*$ is the number of halfplanes
containing $x$ among all halfplanes $h_j^+ \cap h_i$ for all $j \leq i$
and all halfplanes $h_{j'}^- \cap h_i$ for all $j' > i$.
Agarwal et al. showed that this sequence can be computed in $O(n\log n)$ time.

The algorithm considers each plane in $S^*$ one by one in this
order and computes a convex polygon $K_h$ for each plane $h$, which will be defined below.
For simplicity, we let $K_j = K_{h_j}$ for any $j\in\{1,\ldots,n\}$.
For $h_1$, the convex hull of the level $\ell$ in the arrangement of
all lines in $\{h_1 \cap h_i \mid 1 < i \leq n\}$ satisfies the property
for $K_1$.  So, let $K_1$ be the convex hull of the level $\ell$.  The algorithm computes
$K_1$ in $O(n\log^4 n)$ time using the algorithm in~\cite{center-region-2d}.

Now, suppose that we have handled all planes $h_1,\ldots,h_{j-1}$ and
we have $K_1,\ldots,K_{j-1}$ for some $j$.  Let $\Gamma_j = \{K_i \cap h_j \mid 1\leq
i < j\}$. Note that each element in $\Gamma_j$ is a line segment, a ray, or a line. 
Then $K_j$ is defined to be $\ch(\ch(\Gamma_j) \cup (\Lambda_\ell \cap h_j))$.
The set $\ch(\Gamma_j) \cup (\Lambda_\ell \cap h_j)$ 
consists of at most two connected components.  Using 
this property, they give a procedure to compute the intersection
of $K_j$ with a given line segment without knowing $K_j$ explicitly.  More precisely,
they give the following lemma.
Using this procedure and a cutting in $\mathbb{R}^3$, they compute $K_j$ in $O(n^{1+\epsilon})$ time.
Let $H_j = \{h_i \cap h_j \mid 1\leq i < j\}$.

\begin{lemma}[Lemma 2.11. \cite{center-region-3d}]
  \label{lem:procedure}
  Given a triangle $\triangle \subset h_ j$, the set $Z_j$ of edges of
  $K_j$ that intersect the boundary of $\triangle$,
  a segment $e \subset \triangle$, the subset $G \subset H_j$ of the
  $m$ lines that intersect $\triangle$, and an integer $u < m$ such that
  the level $u$ of the arrangement of $G$ coincides with
  $\Lambda_\ell$ within $\triangle$, we can compute the edge of $K_j$ intersecting
  $e$ in $O(m \log^3 (m+|\ch(\Gamma_j)|))$ time.\footnote{
    The running time of the procedure in~\cite{center-region-3d} appears as
    $O(m \polylog (m+|\ch(\Gamma_j)|))$ time. We
    provide a tighter bound.}
\end{lemma}

Here $|\ch(\Gamma_j)|$ is the complexity (that is, the number of edges) of
$\ch(\Gamma_j)$.
Denote this procedure by $\intersection(e, Z_j, \triangle, G, u)$.
By applying this procedure with input satisfying the assumption in
the lemma, we can obtain the edge of $K_j$ intersecting a segment $e\subset \triangle$.

The algorithm computes all edges of $K_j$ using this procedure.
It recursively subdivides the plane
$h_j$ into a number of triangles using $1/r$-nets.
Assume that the followings are given: a triangle $\triangle$, a set $G$ of lines in $H_j$
intersecting $\triangle$, a set $Z_j$ of edges of $K_j$ intersecting the boundary of $\triangle$, and an integer $u$
such that the level $u$ of the arrangement of $G$ coincides with
$\Lambda_\ell$ within $\triangle$.
They are initially set to a (degenerate) triangle $\triangle = h_j$, a
set $G=H_j$, an empty set $Z_j$, and an integer $u = \ell$.
Using this information, the algorithm computes $\triangle \cap K_j$ recursively in 
$O(m^{1+\epsilon})$ time, where $m$ is the size of $H_j$.

Consider the set system $(G, \{\{h \in G \mid h \cap \tau \neq \phi\} \mid
\tau \textnormal{ is a triangle}\})$.  Let $r \in \mathbb{R}$ be a
sufficiently large number.  The algorithm computes a $1/r$-net of $G$ of
size $O(r\log r)$ and triangulates every cell in the arrangement of the $1/r$-net
restricted to $\triangle$.
For each side $e$ of the triangles, the algorithm applies $\intersection(e, Z_j,
\triangle, G, u)$.  Then partial information of $K_j$ is obtained.
Note that $\triangle'$ does not intersect the boundary of $K_j$ for a
triangle $\triangle'$ none of whose edge intersects the boundary of
$K_j$.  Therefore, it is sufficient to consider triangles
some of whose edges intersect the boundary of $K_j$ only.  There are
$O(r\log r\cdot\alpha(r\log r))$ such triangles, where $\alpha(\cdot)$ is the inverse Ackermann function.

Moreover, for each such triangle $\triangle'$, it is sufficient to
consider the lines in $G$ intersecting $\triangle'$. Let $G'$ be the
set of all lines in $G$ intersecting $\triangle'$.  A line lying above
$\triangle'$ does not affect the level of a point in $\triangle'$, so
we do not need to consider it. Thus, the level $u'$ of the arrangement of $G'$ coincides with $\Lambda_\ell$
within $\triangle'$, where $u'$ is $u$ minus the number of lines in $G$ lying
below $\triangle'$.  The following lemma summarizes this argument.
\begin{lemma}
  \label{lem:triangle}
  Consider a triangle $\triangle$ and a set $G$ of lines such that the
  level $u$ of the arrangement of $G$ coincides with $\Lambda_\ell$
  within $\triangle$.  For any triangle $\triangle' \subset \triangle$, 
  the level $u'$ of the arrangement of $G'$ coincides with
  $\Lambda_\ell$ within $\triangle$, where $u'$ is $u$ minus the
  number of lines in $G$ lying below $\triangle'$ and $G'$ is the set
  of lines in $G$ intersecting $\triangle'$.
\end{lemma}

By recursively applying this procedure, the algorithm obtains $K_j \cap
\triangle$.  Therefore, it obtains $K_j$ because $\triangle$ is initially set to $h_j$.

For the analysis of the time complexity, let $T(m, \mu)$
be the running time of the subproblem within $\triangle$, where $m$ is
the number of lines interesting $\triangle$ in $G$ and $\mu$ is the
number of vertices of $K_j$ lying inside $\triangle$.  Then 
the following recurrence inequality is obtained.
$$
T(m,\mu) \leq \sum_{\triangle'} T(\frac{m}{r}, \mu') + O(m\log^3 (m+|\ch(\Gamma_j)|)
+ \mu)
$$
for $m \geq Ar\log r$, where $A$ is some constant independent of $r$.
This inequality holds because the number of lines in $G$ intersecting
a triangle $\triangle'\subset\triangle$ obtained from the arrangement of the $1/r$-net 
is $O(m/r)$ by the property of $1/r$-nets.

It holds that $T(m,\mu) = O(m^{1+\epsilon}\log^3 (m+|\ch(\Gamma_j)|))$ for
any constant $\epsilon > 0$.  Initially, $m$ and $|\ch(\Gamma_j)|$ are $O(n)$.
Thus the overall running time for each plane in $S^*$ is $O(n^{1+\epsilon})$. Therefore,
$K_j$ can be computed in $O(n^{2+\epsilon})$ time in total for all $1\leq j\leq n$,
and the convex hull of $\Lambda_\ell$ can be computed in the same time.

\begin{theorem}[Theorem 2.10. \cite{center-region-3d}]
  Given a set $S$ of $n$ points in $\mathbb{R}^3$ and an integer $\ell$, the set of points of Tukey depth at least $\ell$ with respect to $S$ can be computed
  in $O(n^{2+\epsilon})$ time for any constant $\epsilon>0$.
\end{theorem}

\subsection{Our Algorithm}
In this subsection, we show how to compute $K_j$ for an integer $1
\leq j \leq n$ in $O(n \log^4 n)$ time. This leads to the total running
time of $O(n^2 \log^4 n )$ by replacing the corresponding procedure 
of the algorithm in~\cite{center-region-3d}. 
Recall that the previous algorithm considers the triangles in the triangulation of
the arrangement of an $1/r$-net.  Instead, we consider finer triangles.

Again, consider a triangle $\triangle$, which is initially set to the
plane $h_j$.  We have a set $G$ of lines,  which is initially set to
$H_j=\{h_i\cap h_j\mid 1\leq i<j\}$,
and an integer $u$, which is initially set to $\ell$.  We compute a
$1/r$-net of the set system defined on the lines in $H_j$
intersecting $\triangle$ as the previous algorithm does.  Then we
triangulate the cells in the arrangement of the $1/r$-net.

\begin{figure}
  \begin{center}
    \includegraphics[width=0.9\textwidth]{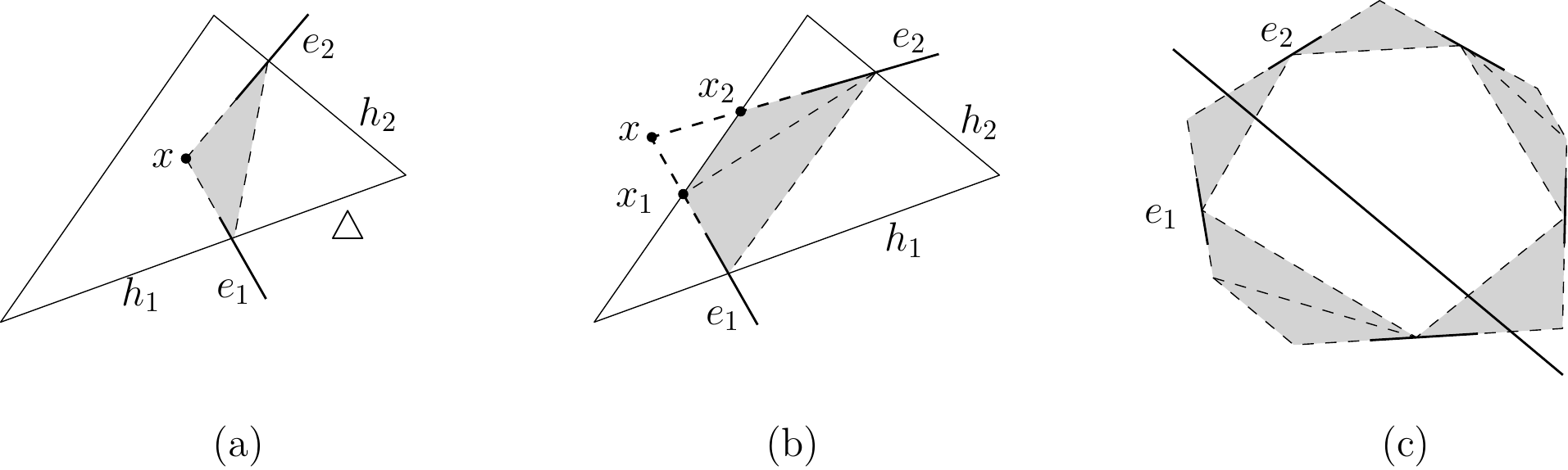}
		\caption{\small
			(a) The gray triangle is a triangle we obtained from $e_1$ and $e_2$.
			(b) If $x$ lies outside of $\triangle$, we obtain two triangles.
			(c) The gray region is the convex hull of all triangles we obtained. 
			Any line segment intersects at most four triangles in $E'$.
			\label{fig:subdivide}}
                    \end{center}
                  \end{figure}

For each edge $e$ of the triangulation of the cells of the arrangement of the $1/r$-net, 
we compute the edge of $K_j$ intersecting $e$ by applying $\intersection(e,Z_j, \triangle, G,u)$.
Let $E$ be the set of triangles in the triangulation at least one of whose sides intersects
$K_j$, and $K$ be the list of edges of $K_j$ intersecting triangles in
$E$ sorted in clockwise order along the boundary of $K_j$. 
We can compute $K$ although we do not know the whole description of $K_j$
because $K_j$ is convex.
Note that a triangle in $E$ is crossed by an edge of $K_j$, or intersected by two consecutive edges of $K$.
The previous algorithm applies this procedure again for the
triangles in $E$.  But, our algorithm subdivides the triangles in $E$
further.

For two consecutive edges $e_1$ and
$e_2$ in $K$, let $x$ be the intersection point of two lines, one containing
$e_1$ and one containing $e_2$.  See Figure~\ref{fig:subdivide}(a).  
Both $e_1$ and $e_2$ intersect a common triangle $\triangle$ in
$E$.  Let $h_1$ and $h_2$ be sides of $\triangle$ intersecting
$e_1$ and $e_2$, respectively. The two sides might coincide with each other.

If $x$ is contained in $\triangle$, then we consider the triangle
with three corners $x$, $e_1\cap h_1$, and $e_2\cap h_2$. See
Figure~\ref{fig:subdivide}(a).  If $x$ is not contained in
$\triangle$, let $x_1$ be the intersection point of the line containing
$e_1$ with the side of $\triangle$ other than $h_1$ and $h_2$. See
Figure~\ref{fig:recursion}(b).  Similarly, let $x_2$ be the
intersection point of the line containing $e_2$ with the side of
$\triangle$ other than $h_1$ and $h_2$.  In this case, we consider
two triangles: the triangle with corners $e_1\cap h_1$, $x_1$,
$e_2\cap h_2$ and the triangle with corners $x_1$, $x_2$,
$e_2\cap h_2$.

Now, we have one or two triangles for each pair of two consecutive
edges in $K$.  Let $E'$ be the set of such triangles.  By
construction, the union of all triangles in $E'$ contains the boundary
of $K_j$.  Thus, we can compute the
intersection of the boundary of $K_j$ with $\triangle$ by applying this procedure
recursively within $\triangle$. 

For each triangle in $E'$, we compute the intersection of the boundary
of $K_j$ with the triangle recursively as the previous algorithm does.
For each triangle $\triangle' \in E'$, we define $G(\triangle')$ to be
the set of lines in $G$ intersecting $\triangle'$. And we define $u'$
to be $u$ minus the number of line segments lying below $\triangle'$.
By Lemma~\ref{lem:triangle}, the level $u'$ of the arrangement of
$G(\triangle')$ coincides with $\Lambda_\ell$ within $\triangle$.  We
compute the intersection of $K_j$ with the sides of each triangle in
$E'$ by applying $\intersection$ procedure.

Now, we analyze the running time of our algorithm.  The following
technical lemma and corollary allow us to analyze the running time.
For an illustration, see Figure~\ref{fig:subdivide}(c).
\begin{lemma}
  \label{lem:four}
  A line intersects at most four triangles in $E'$.
\end{lemma}
\begin{proof}
  By construction, the union of all triangles in $E'$ coincides with
  $C_1\setminus C_2$ for a convex polygon $C_1$ containing all edges of $K$ on its boundary
  and a convex polygon $C_2$ whose vertices are on edges of $K$.  See
  Figure~\ref{fig:subdivide}(c). 
  A connected component of  $\textsf{int}(C_1)\setminus \textsf{int}(C_2)$ is the union of (one or two) finer triangles in $E'$ 
  obtained from a single triangle in $E$, where
  $\textsf{int}(C)$ is the interior of a convex polygon $C$.
  Since a line intersects at most two connected components of
  $\textsf{int}(C_1)\setminus \textsf{int}(C_2)$, a line intersects at most four triangles in $E'$.
\end{proof}
\begin{corollary}
  \label{corollary:linear}
  The total sum of the numbers of lines in $G(\triangle')$ over all triangles
  $\triangle' \in E'$ is four times the number of lines in $G$.
\end{corollary}

We iteratively subdivide $h_j$ using $1/r$-nets until we obtain $K_j$.
Initially, we consider the whole plane $h_j$, which intersects at most
$n$ lines in $H_j$.  In the $i$th iteration, each triangle we consider
intersects at most $n/r^i$ lines in $H_j$ by the property of $1/r$-nets.
This means that in $O(\log_r n)$ iterations, every triangle intersects
a constant number of lines in $H_j$. Then we stop subdividing the plane
and compute $K_j$ lying inside each triangle in constant time.

Consider the running time for each iteration.  For each triangle, we
first compute a $1/r$-net in time linear to the number of lines
intersecting the triangle.  Then we apply $\intersection$ procedure of
Lemma~\ref{lem:procedure} for each edge in the arrangement of the
$1/r$-net. This takes $O(m\log^3 m/r^2)$ time, where $m$ is the number
of lines in $H_j$ intersecting the triangle.

In each iteration, we have $O(n)$ triangles in total because every triangle
contains at least one vertex of $K_j$.  Moreover, the sum of the
numbers of lines intersecting the triangles is $O(n)$ by
Corollary~\ref{corollary:linear}.  This concludes that the running
time for each iteration is $O(n \log^3 n)$.

Since we have $O(\log_r n)$ iterations, we can compute $K_j$ in
$O(n \log^4 n)$ time.  Recall that the convex hull of the level $\ell$
of the arrangement of the planes is the convex hull of $K_j$'s for
all indices $1\leq j\leq n$.  Therefore, we can compute the level
$\ell$ in $O(n^2 \log^4 n)$ time, and compute the set of points of
depth at least $\ell$ in the same time.

\begin{theorem}
  Given a set of $n$ points in $\mathbb{R}^3$ and an integer
  $\ell \geq 0$, the set of points of depth at least $\ell$ can be
  computed in $O(n^2\log^4 n)$ time.
\end{theorem}

\section{Computing the Colorful Center Region in \texorpdfstring{$\mathbb{R}^2$}{R2}}
\label{sec:second-problem}
In this section, we consider the colored version of the Tukey depth in
$\mathbb{R}^2$. Let $\ell>0$ be an integer at most $n$.  
We use integers from 1 to $k\in\mathbb{N}$ to represent
colors. Let $S$ be a
set of $n$ points in $\mathbb{R}^2$ each of which has exactly one
color among $k$ colors. We assume that for each color $i$, there exists at least one point of color $i$
in $S$.

The definitions of the Tukey depth and center region are extended to this
colored version. The \emph{colorful (Tukey) depth} of a point $x$ in $\mathbb{R}^2$ is defined as 
the minimum number of different colors of points contained in a closed 
halfspace containing $x$.
The \emph{colorful center region} is the set of points in $\mathbb{R}^2$ whose colorful Tukey depths
are at least $\lceil k/(d+1) \rceil$.
A \emph{colorful Tukey median} is a point in $\mathbb{R}^2$ with the largest colorful Tukey depth.

We provide a lower bound of the colorful depth of a
colorful Tukey median, which is analogous to properties of the
standard Tukey depth.  The proof is similar to the one for the standard Tukey depth.
\begin{lemma}
  \label{lem:lower_bound}
  The colorful depth of a colorful Tukey median is at least $\lceil
  k/(d+1) \rceil$. Thus, the colorful center region is not empty.
\end{lemma}
\begin{proof}
	Let $S$ be a set of $n$ colored points in $\mathbb{R}^d$ each with one of $k$
	colors. We choose any $k$ points in $S$ with distinct colors and denote the set
	of the $k$ points by $S'$.	Let $c$ be a Tukey median of $S'$. 
	The Helly's theorem implies that the depth of $c$ with respect to $S'$ is at least $\lceil k/(d+1) \rceil$.
	By definition, the colorful depth of $c$ with respect to $S$ 
	is at least $\lceil k/(d+1) \rceil$.
\end{proof}

In this section, we present an algorithm to compute the set of all
points of colorful depth at least $\ell$ with respect to $S$ in $\mathbb{R}^2$.  By
setting $\ell=\lceil k/(d+1) \rceil$, we can compute the center region
using this algorithm.
Our algorithm follows the approach in~\cite{center-point-2d}.

\subsection{Duality Transform}
As the algorithm in~\cite{center-point-2d} does, we use a duality of points and lines.  Due to the properties
of the duality, our problem reduces to computing the convex hull of a level of the
arrangement of $x$-monotone polygonal curves.

The standard duality transform maps a point $s$ to a line $s^*$, and a
line $h$ to a point $h^*$.  Let $S^* = \{s^* \mid s \in S\}$. Each line
$s^*$ in $S^*$ has the same color as $s$.  We consider the
colorful depth of a point in the dual space.  We define a
\emph{colorful level} of a point $x$ in the dual plane 
with respect to $S^*$ to be the number of different colors of lines lying below $x$ or
containing $x$.  A point $x$ in the primal plane 
with respect to $S$ has colorful depth at least $\ell$ if and only if all points in the
line $x^*$ have colorful level at least $\ell$ and at most $k-\ell$ in the dual plane.

Note that for a color $i\in\{1,\ldots,k\}$, a line in $S^*$ with color $i$ lies below a point $x$ 
if and only if $x$ lies above the lower envelope of lines in $S^*$ of color
$i$ in the dual plane.  With this property, we can give an alternative definition of the
colorful level.  For each color $i$, we consider the lower envelope
$C_i$ of all lines in $S^*$ of color $i$.  The colorful level of a
point $x \in \mathbb{R}^2$ is the number of different lower envelopes $C_i$
lying below $x$ or containing $x$.  In other words, the colorful level
of a point with respect to $S^*$ is the level of the point with
respect to the set of the lower envelopes $C_i$ for $i=1,\ldots,k$.

In the following, we consider the arrangement of the lower
envelopes $C_i$ for $i=1,\ldots,k$.  In this case, a cell in this arrangement is not
necessarily convex.  Moreover, this arrangement does not satisfy the
property in Lemma~4.1 of~\cite{center-region-2d}.  Thus, the algorithm
in~\cite{center-region-2d} does not work for the colored version
as it is.

Let $L_\ell$ be the set of points in the dual plane of colorful level
at most $\ell$.  Similarly, let $U_\ell$ be the set of points in
the dual plane of colorful level at least $\ell$. By definition, the dual line of a
point of colorful depth at least $\ell$ is contained neither in $L_\ell$ nor
in $U_{k-\ell}$ in the dual plane.  Moreover, such a line lies outside of both the
convex hull $\ch(L_\ell)$ of $L_\ell$ and the convex hull $\ch(U_{k-\ell})$ of $U_{k-\ell}$.

Thus, once we have $\ch(L_\ell)$  and $\ch(U_{k-\ell})$,
we can compute the set of all points of colorful depth
at least $\ell$ in linear time.  In the following subsection, we show how to
compute $\ch(L_\ell)$ in $O(n\log^4 n)$ time.  The convex hull of $U_{k-\ell}$
can be computed analogously.

\subsection{Computing the Convex hull of \texorpdfstring{$L_\ell$}{L}}
In this subsection, we present an algorithm for computing the convex hull $\ch(L_\ell)$
of $L_\ell$. We subdivide the plane into $O(n)$ vertical slabs and compute
$\ch(L_\ell)$
restricted to each vertical slab in Section~\ref{sec:main}.
To do this, we compute the intersection between $\ch(L_\ell)$ 
and a vertical line defining a vertical slab. This subprocedure is described in Section~\ref{sec:color-subprocedure}.

\subsubsection{Subprocedure: Computing the Intersection of the Convex Hull with a Line}
\label{sec:color-subprocedure}
Let $h$ be a vertical line in $\mathbb{R}^2$.  In this subsection, we
give a procedure to compute the intersection of $h$ with $\ch(L_\ell)$.
This procedure is used as a subprocedure of the
algorithm in Section~\ref{sec:main}.  We modify the procedure
by Matous\v{e}k~\cite{center-region-2d}, which deals with the standard
(noncolored) version of the problem.

The following lemma gives a procedure for determining whether a point $x$ on $h$ lies above
$\ch(L_\ell)$ or not. This procedure is used as a subprocedure in the
procedure for computing the intersection of $h$ and $\ch(L_\ell)$ described in Lemma~\ref{lem:intersection}.

\begin{lemma}
  \label{lem:first-level}
  We can determine in $O(n\log n + k\log^2 n)$ time whether a given
  point $x$ lies above the convex hull of $L_\ell$.  In addition, we
  can compute the lines tangent to $\ch(L_\ell)$ passing through $x$
  in the same time.
\end{lemma}
\begin{proof}
By definition, $L_\ell$ is $x$-monotone and $\ch(L_\ell)\cap g$ is
a ray (halfline) going vertically downward for any vertical line $g$.
A point $x$  lies above $\ch(L_\ell)$ if and only if there is a line passing
through $x$ and tangent to $\ch(L_\ell)$.
Note that there are exactly two such tangent lines for a point $x$ lying above
$\ch(L_\ell)$: one is tangent to $\ch(L_\ell)$ at a point lying left to $x$, and the other is
tangent to $\ch(L_\ell)$ at a point lying right to $x$.
Among them, we show how to compute the line tangent to $\ch(L_\ell)$ at a
point lying left to $x$.  The other line can be computed
analogously.  To apply parametric search, we give a
decision algorithm to check whether a given ray $\gamma$ starting from 
$x$ lies above $\ch(L_\ell)$ or not.

\paragraph{Checking whether a ray \texorpdfstring{$\gamma$}{r} lies above the convex hull.}
We first compute the intersection points of $\gamma$ with $C_i$ for each
$i$, and sort them along
$\gamma$.  Recall that $C_i$ is the lower envelope of lines in
$S^*$ of color $i$.  Since each $C_i$ is a convex
polygonal curve, the total number of intersection points is at most $2k$.  The
intersection points can be computed and sorted in $O(k\log n + k\log
k)=O(k \log n)$ time.  We walk along $\gamma$ from a point at infinity and compute
the colorful level of each intersection point one by one.  We can
compute the intersection points of $L_\ell$ with $\gamma$ in $O(k)$ time, and compute
the intersection points of $\ch(L_\ell)$ with $\gamma$ in the same time.

\begin{figure}
	\begin{center}
		\includegraphics[width=0.35\textwidth]{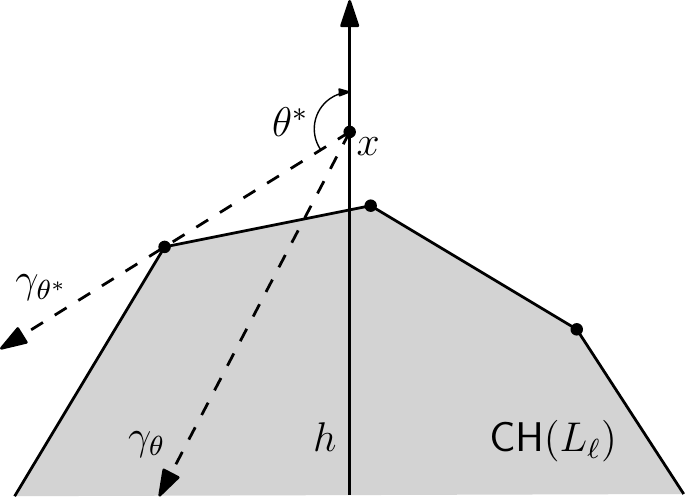}
		\caption{\small The ray $\gamma_{\theta^*}$ is tangent to $\ch(L_\ell)$.
			For any $\theta \leq \theta^*$, the ray $\gamma_{\theta}$ intersects $\ch(L_\ell)$.
			But for any $\theta^* < \theta < \pi$, the ray $\gamma_{\theta^*}$ does not
			intersect $\ch(L_\ell)$.
			\label{fig:first-level}}
	\end{center}
\end{figure}

\paragraph{Applying parametric search.}
We use this decision algorithm to check whether there is a
line passing through $x$ and tangent to $\ch(L_\ell)$.  For an
angle $\theta \in [0, \pi)$, let $\gamma_\theta$ denote the ray
starting from $x$ such that the clockwise angle from
$\gamma_\theta$ to the $y$-axis (towards the positive
direction) is $\theta$.  See
Figure~\ref{fig:first-level}.  Let $\theta^*$ be the angle such
that $\gamma_{\theta^*}$ is tangent to $\ch(L_\ell)$ at some point
lying left to $x$.  Then, for any angle $\theta \in
[0,\theta^*]$, the ray $\gamma_\theta$ does not intersect $\ch(L_\ell)$.
For any angle $\theta \in (\theta^*,\pi)$, the ray $\gamma_\theta$
intersects $\ch(L_\ell)$.

Initially, we have an interval $[0,\pi]$ for
$\theta^*$.  In the following, we reduce the 
interval, which contains $\theta^*$, until we find $\theta^*$.
Consider the vertices of $C_i$ lying to the left of $x$ for every
$i$.  We compute all angles $\theta'$ such that $\gamma_{\theta'}$
goes through a vertex of $C_i$ for some $i$ and sort them in the
increasing order.  Let $\theta_1,\ldots,\theta_{n'}$ be the
angles in the increasing order. Note that $n' \leq n$.  We apply binary search on
these angles using the decision algorithm to find the interval $I=(\theta_j, \theta_{j+1})$
containing $\theta^*$ for some $1\leq j\leq n'$ in
$O(n \log n+k\log^2 n )$ time.

Recall that the lower envelopes $C_i$'s (for $i=1,\ldots,k$)
consist of $O(n)$ line segments in total.  For any angle
$\theta \in I$, the set of the line segments intersecting
$\gamma_\theta$ remains the same.  But the order of the
intersection points of such line segments with $\gamma_\theta$ along
$\gamma_\theta$ may change over angles in $I$.  Now, we will find an interval $I' \subset
I$ containing $\gamma_{\theta^*}$ such that the order of the
intersection points remains the same for any $\theta \in I'$.

Let $\mathcal{C}$ be the set of the line segments of $C_i$'s intersected by
$\gamma_\theta$ for some $\theta \in I$. Since $C_i$'s are convex,
the size of $\mathcal{C}$ is at most $2k$. We sort the line segments
in $\mathcal{C}$ along $\gamma_{\theta^*}$ without explicitly
computing $\gamma_{\theta^*}$ as follows.  To sort the line segments, we need
to determine the order of two line segments, say $s_1$ and
$s_2$, along $\gamma_{\theta^*}$.  We can do this using the
decision algorithm.  If the two line segments do not intersect
each other, we can compute the order of them directly since the order of them
remains the same over angles in $I$.
Otherwise, let $s$ be the intersection point of $s_1$ and $s_2$, and
$\theta_s$ be the angle such that $\gamma_{\theta_s}$ intersects
$s$.  If $\theta_s \notin I$, we can compute the order of the
two line segments directly.  If not, we apply the decision
algorithm with input $\theta_s$. The decision algorithm
determines whether $\theta^* \geq \theta_s$ or not. So, we can
reduce the interval $I$ and determine the order of $s_1$ and
$s_2$.

We need $O(k\log k)$ comparisons to sort $2k$ elements. To
compare two line segments, we apply the decision
algorithm. Thus, the running time is $O(k\log k \cdot (k\log^2
n+n\log n))$.  But we can reduce the running time by using the
parallel sorting algorithm described in~\cite{parallel-sorting}.  The
parallel sorting algorithm consists of $O(\log k)$ iterations,
and each iteration consists of $O(k)$ comparisons which are
independent to the others.  In each iteration, we compute the
angles corresponding to the comparisons. We have $O(k)$ angles
and sort them in the increasing order. We apply binary search
using the decision algorithm.  Then, after applying the
decision algorithm $O(\log k)$ times, we can complete the
comparisons in the iteration.  Thus, in total, the algorithm
takes $O(k\log n\log^2 k)$ time.
Moreover, we can reduce the running time further by applying
an extension to Megiddo's technique due to
Cole~\cite{cole-parametric}.  The running time of the
algorithm is $O(k \log n\log k)$.

Since we compute $I$ in $O(n\log n+k\log^2 n)$ time and compute $I'$ in $O(k\log n\log k)$ time,
the overall running time is $O(n\log n+k\log^2 n)$.
\end{proof}

Using Lemma~\ref{lem:first-level} as a subprocedure, we can compute
for a given vertical line $h$ the intersection of $h$ with $\ch(L_\ell)$.
\begin{lemma}
  \label{lem:intersection}
  Given a vertical line $h$, the intersection of $h$ with $\ch(L_\ell)$
  can be computed in $O(n\log^2 n+ k\log^3 n)$ time.
\end{lemma}
\begin{proof}
We again apply parametric search on the line $h$ to find the
intersection point $x^*$ of $h$ with $\ch(L_\ell)$.  Initially, we
set an interval $I = h$. We will reduce the 
interval in three steps. In every step, the interval contains $x^*$.

\paragraph{The first step.}
Let $\mathcal{C}$ be the set of all edges (line segments) of
lower envelopes $C_i$ for all $i$.  For each line segment $s$
in $\mathcal{C}$, we denote the line containing $s$ by
$\hat{s}$.  We compute the intersection points between $h$ and
$\hat{s}$ for every line segment $s$, and sort them along $h$.
We apply binary search on the sorted list of the intersection points using the
algorithm in Lemma~\ref{lem:first-level}.  Then we have the
lowest intersection point $x_1$ that lies above $\ch(L_\ell)$
and the highest intersection point $x_2$ that lies below
$\ch(L_\ell)$.  They can be computed in $O(n\log^2 n+ k\log^3
n)$ time.  Note that $x^*$ lies between $x_1$ and $x_2$ along $h$.  We
let $I$ be the interval between $x_1$ and $x_2$.

\paragraph{The second step.}
We reduce the interval $I$ containing $x^*$ such that
for any point $x \in I$, the set of line segments in
$\mathcal{C}$ intersected by $\gamma_x$ remains the same, where $\gamma_x$
is the ray starting from $x$ and tangent to $\ch(L_\ell)$ at
a vertex of $\ch(L_\ell)$ lying left to $x$.  To this end, for every endpoint
$c$ of the line segments in $\mathcal{C}$, we check whether
$c$ lies above $\gamma_{x^*}$ or not.  This can be
done in $O(k\log n)$ time for each endpoint by applying the decision algorithm in the proof of 
Lemma~\ref{lem:first-level}.  Although we
have $O(n)$ endpoints of line segments in $\mathcal{C}$, we do
not need to apply the decision algorithm in the proof of Lemma~\ref{lem:first-level}
for all of them.

For each endpoint $c$ of the line segments in $\mathcal{C}$
lying left to $h$, we consider its dual $c^*$.  Let
$\mathcal{C}'$ be the set of the lines which are dual of the 
endpoints of the line segments in $\mathcal{C}$. 
Let $\gamma^*$ be the point dual to the line containing ray $\gamma_{x^*}$.
For a line $c^*$, we can determine whether $\gamma^*$ lies above $c^*$
or not in $O(k\log n)$ time using the decision algorithm in the proof of 
Lemma~\ref{lem:first-level}.  In the
  arrangement of the lines $c^*$ for all endpoints $c$ of the line segments,
  we will find the cell that contains $\gamma^*$  without constructing the arrangement explicitly. as follows.

For a parameter $r$ with $1\leq r\leq n$, 
a $(1/r)$-cutting of $\mathcal{C}'$ is defined as 
a set of interior-disjoint (possibly unbounded) triangles 
whose union is the plane with the following property: 
no triangle of the cutting is intersected by more than $n/r$ lines in $\mathcal{C}'$.
We compute a $(1/r)$-cutting of size $O(r^2)$ in $O(nr)$ time using the algorithm by Chazelle~\cite{chazelle1993}.
The number of lines in $\mathcal{C}'$ intersecting $\triangle$ is
$|\mathcal{C}'|/r$ for any $\triangle$ in the cutting.  Note that for a line $c^*$ in
$\mathcal{C}'$ which does not intersecting $\triangle$, we can
check whether $\gamma^*$ lies above $c^*$ or not in constant time.
For each $\triangle$ in the cutting, we check whether it contains $\gamma^*$ 
in $O(k\log n)$ time using the decision algorithm in the proof of  Lemma~\ref{lem:first-level}.
Note that there is exactly one triangle in the cutting containing $\gamma^*$.
We apply the $(1/r)$-cutting within the triangle recursively
until we find the cell in the arrangement of
$\mathcal{C}'$ containing $\gamma^*$. This can be done in $O(n+k\log^2 n)$ time.  We have the interval $I$
for $x^*$ with the desired property.

\paragraph{The third step.}
Let $\mathcal{C}_I$ be the set of line segments in
$\mathcal{C}$ intersecting $\gamma_x$ for every $x \in I$. Note that the size of $\mathcal{C}_I$ 
is at most $2k$. Recall
that $\gamma_x$ is the ray starting from $x$ and tangent to
$\ch(L_\ell)$ at some point lying left to $x$. In this step, our goal is to sort the
line segments in $\mathcal{C}_I$ along the ray $\gamma_{x^*}$
without explicitly computing $x^*$.  To sort them, for two
line segments $s_1$ and $s_2$ in $\mathcal{C}_I$, we need to
determine whether $s_1$ comes before $s_2$ along $\gamma_{x^*}$ as
follows.  We compute the intersection point $s$ between $s_1$ and $s_2$
and check whether $\gamma_{x^*}$ lies above $s$ or not using the decision algorithm in the proof of 
Lemma~\ref{lem:first-level}.  Then we can determine the order
for $s_1$ and $s_2$ along $\gamma_{x^*}$ because we have already
reduced the interval in the first step.  (If the
intersection does not exist, we can determine the order directly.)
To sort all line segments efficiently, we again use Cole's
parallel sorting algorithm and a cutting as we did in the second step and in Lemma~\ref{lem:first-level}. 
Then we can sort all line segments in $O(k\log^2 k\log n)$ time.

\paragraph{Computing the intersection point.}
We have the interval $I$ containing $x^*$ such that the order of line segments in $\mathcal{C}$ 
intersecting $\gamma_x$ remains the same for any $x\in I$. 
Notice that the procedure in Lemma~\ref{lem:first-level} depends
only on this order. Thus any point $x\in I$ lies above the convex hull of $L_\ell$.
Therefore, the lowest point in $I$ is $x^*$ by definition.
In total, we can compute $x^*$ in $O(n\log^2 n+ k\log^3 n)$ time.
\end{proof}

\subsubsection{Main Procedure: Computing the Convex hull of \texorpdfstring{$L_\ell$}{L}}
\label{sec:main}
We are given a set of $k$ polygonal curves (lower envelopes) of total
complexity $O(n)$ and an integer $\ell$.  Let $\mathcal{C}$ be the set
of the edges (line segments) of the $k$ polygonal curves.
Recall that $L_\ell$ is the set of points of colorful level at most
$\ell$.  That is, $L_\ell$ is the set of points lying above (or
contained in) at most $\ell$ polygonal curves.  In this subsection, we
give an algorithm to compute the convex hull of $L_\ell$.

Basically, we subdivide the plane into $O(n)$ vertical slabs such that
each vertical slab does not contain any vertex of the
$k$ polygonal curves in its interior.  We say a vertical slab is \emph{elementary} if
its interior contains no vertex of the $k$ polygonal curves.  Let $A$
be an elementary vertical slab of the subdivision and $\mathcal{Q}$ be
the set of the intersection points of the line segments in $\mathcal{C}$ with $A$.

\begin{figure}
  \begin{center}
    \includegraphics[width=0.8\textwidth]{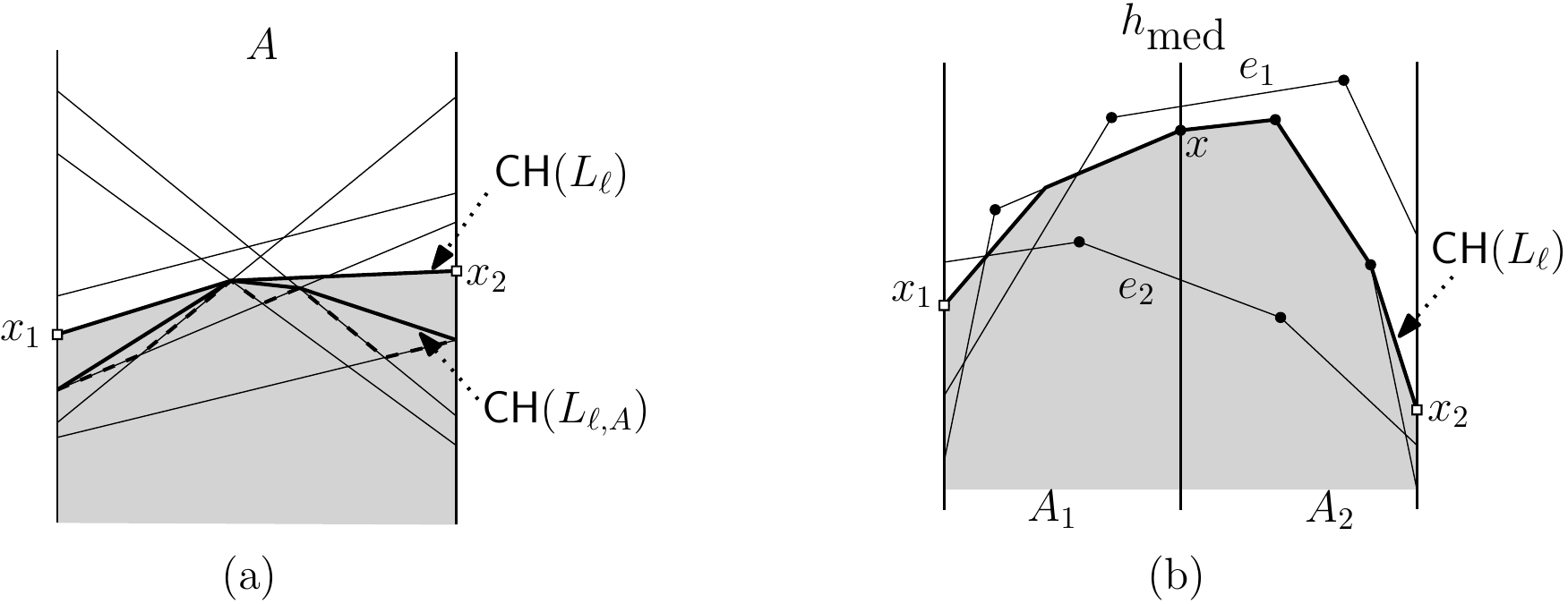}
 		\caption{\small (a) $\ch(L_\ell) \cap A$ coincides with the convex hull of $x_1$, $x_2$, and 
 		$\ch(L_{\ell,A})$. (b) We put $e_1$ only to $\mathcal{Q}_2$ and $e_2$ to both
 		sets.
 	\label{fig:recursion}}
  \end{center}
\end{figure}

Consider the arrangement of the line segments in
$\mathcal{Q}$ restricted to $A$.  See
Figure~\ref{fig:recursion}(a).  Let $\ch(L_{\ell,A})$
denote the convex hull of points of level at most
$\ell$ in this arrangement.  Note that
$\ch(L_{\ell,A})$ is contained in $\ch(L_\ell) \cap
A$, but it does not necessarily coincide with $\ch(L_\ell) \cap
A$.

By the following observation, we can compute
$\ch(L_\ell) \cap A$ once we have $\ch(L_{\ell,A})$.

\begin{observation}
  The intersection of $\ch(L_\ell)$ with $A$ coincides
  with the convex hull of $x_1$, $x_2$, and $\ch(L_{\ell,A})$, where
  $x_1$ and $x_2$ are the intersection points of $\ch(L_\ell)$ with the
  vertical lines bounding $A$.
\end{observation}

A subdivision of $\mathbb{R}^2$ into elementary vertical slabs with respect
to the $k$ polygonal curves
can be easily computed by taking all vertical lines passing through endpoints of the line
segments in $\mathcal{C}$.  However, the total complexity of
$\mathcal{Q}$ over all slabs $A$ is $\Omega(n^2)$.
Instead, we will choose a subset $\mathcal{Q}'$
of $\mathcal{Q}$ and a value $\ell'$ such that $L_\ell \cap A$
coincides with the level $\ell'$ with respect to
$\mathcal{Q}'$, and the total complexity of $\mathcal{Q}'$ is
linear.  In the following, we show how to choose $\mathcal{Q}'$ for
every elementary slab. 

\paragraph{Subdividing the region into two vertical slabs.}
Initially, the subdivision of $\mathbb{R}^2$ is the plane itself.  We
subdivide each vertical slab into two vertical subslabs recursively until
every slab becomes elementary. While we subdivide a slab $A$,
we choose a set $\mathcal{Q}'$ for $A$.

Consider a vertical slab $A$ and assume that we already have
$\mathcal{Q}'$ and a level $\ell'$ for $A$.  Assume further that
we already have the intersection points $x_1$ and $x_2$ of
$\ch(L_\ell)$ with the vertical lines bounding $A$. See Figure~\ref{fig:recursion}(a). 
We find the vertical line $h_{\text{med}}$ passing through the
median of the endpoints of line segments in $\mathcal{Q}'$ with respect
to their $x$-coordinates in $O(|\mathcal{Q}'|)$ time, where $|\mathcal{Q}'|$ is the cardinality of $Q'$.  The vertical line subdivides
$A$ into two vertical subslabs.  Let $A_1$ be the subslab lying left
to the vertical line, and $A_2$ be the other subslab.  We compute
$\ch(L_{\ell',A}) \cap h_\textnormal{med}$ by applying the algorithm
in Lemma~\ref{lem:intersection} in $O(|\mathcal{Q}'|\log^2|\mathcal{Q}'|+k'\log^3 |\mathcal{Q}'|)$ time, where
$k'$ is the number of distinct colors assigned to the line segments in $\mathcal{Q}'$. Since $k'\leq |\mathcal{Q}'|$, this running time is $O(|\mathcal{Q}'|\log^3|\mathcal{Q}'|)$.
We denote the intersection point by $x$.  While computing $x$, we can
compute the slope $\tau$ of the edge of $\ch(L_\ell)$ containing $x$.
(If $x$ is the vertex of the convex hull, we compute the slope
of the edge lying left to $x$.)  See
Figure~\ref{fig:recursion}(b).

We show how to compute two sets $\mathcal{Q}_1,
\mathcal{Q}_2$ and two integers $\ell_1, \ell_2$ such that
$\ch(L_{\ell',A})$ is the convex hull of $x$, $\ch(L_{\ell_1,A_1})$ and
$\ch(L_{\ell_2,A_2})$, where $\ch(L_{\ell_t,A_t})$ is the convex
hull of the level $\ell_t$ in the arrangement of $\mathcal{Q}_t$
for $t=1,2$.  The two sets are initially set to be empty, and two
integers are set to be $\ell'$.  Then we consider each line segment
$s$ in $\mathcal{Q}'$.  If $s$ is fully contained in one subslab
$A_t$, then we put $s$ only to $\mathcal{Q}_t$.
Otherwise, $s$ intersects $h_\textnormal{med}$.  If $s$ lies above
$x$, then we compare the slope of $s$ and $\tau$.  Without loss of
generality, we assume that $\tau \geq 0$.  If the slope of $s$ is
larger than $\tau$, $\ch(L_{\ell',A}) \cap A_2$ does not intersect
$s$. Thus, a point of level at most $\ell'$ in the arrangement of
$\mathcal{Q}'$ restricted to $A_2$ has level at most $\ell'$ in the
arrangement of $\mathcal{Q}' \setminus \{s\}$ restricted to $A_2$.
This means that we do not need to put $s$ to $\mathcal{Q}_2$.  We
put $s$ only to $\mathcal{Q}_1$.  The case that the slope of $s$
is at most $\tau$ can be handled analogously.

Now, consider the case that $s$ lies below $x$.  If both endpoints are
contained in the interior of $A$, we put $s$ to both
$\mathcal{Q}_1$ and $\mathcal{Q}_2$.  Otherwise, $s$ crosses
one subslab, say $A_1$.  In this case, we put $s$ to
$\mathcal{Q}_2$.  For $A_1$, we check whether $s$ lies below the
line segment connecting $x_1$ and $x$.  If so, we set $\ell'_1$ to
$\ell'_1-1$ and do not put $s$ to the set for $A_1$.  This is because
$\ch(L_{\ell',A})$ contains $s$.  Otherwise, we put $s$ to the set for
$A_1$.

We analyze the running time of the procedure.  In the $i$th iteration,
each vertical slab in the subdivision contains at most $n/2^i$
endpoints of the line segments in $\mathcal{C}$.  Thus, we can
complete the subdivision in $O(\log n)$ iterations.

Each iteration takes $O(\sum_{j} n_j\log^3 n_j)$ time, where $n_j$ is
  the complexity of $\mathcal{Q}'_{A_j}$ for the $j$th leftmost slab $A_j$ in the iteration.
By construction, each line segment in $\mathcal{C}$ is contained in at
most two sets defined for two vertical slabs in the same iteration.  Therefore, each
iteration takes $O(\sum_{j} n_j\log^3 n_j)=O(n\log^3 n)$ time.

\paragraph{Computing the convex hull inside an elementary vertical slab.}
\label{sec:elementary}
We have $O(n)$ elementary vertical slabs.
Each elementary vertical slab has a set of line segments, and the total
number of line segments in all vertical slabs is $O(n)$.
For each elementary vertical slab with integer $\ell'$,
we have to compute the convex hull of the level $\ell'$ in the
arrangement of its line segments.

Matous\v{e}k~\cite{center-region-2d} gave an $O(n\log^4 n)$-time
algorithm to compute the convex hull of the level $\ell$ in the
arrangement of lines.  In our problem, we want to compute the convex
hull of the level $\ell$ in the arrangement of lines restricted to a
vertical slab.  The algorithm in~\cite{center-region-2d} works also
for our problem (with modification).  Since this modification is
straightforward, we omit the details of this procedure.
\begin{lemma}
  The convex hull of $L_\ell$ can be computed in $O(n \log^4 n)$ time.
\end{lemma}

\begin{theorem}
  Given a set $S$ of $n$ colored points in $\mathbb{R}^2$ and an
  integer $\ell$, the set of points of colorful depth at most $\ell$
  with respect to $S$ can be computed in $O(n \log^4 n)$ time.
\end{theorem}
\begin{corollary}
  Given a set $S$ of $n$ colored points in $\mathbb{R}^2$, the
  colorful center region of $S$ can be computed in $O(n \log^4 n)$
  time.
\end{corollary}

\section{Computing the Colorful Center Region in \texorpdfstring{$\mathbb{R}^3$}{R3}}
By combining ideas presented in the previous sections, we can compute
the colorful level $\ell$ of a set of $n$ colored points in $\mathbb{R}^3$ in $O(n^2\log^4 n)$ time
for any integer $\ell$.
We use a way to subdivide the planes described in our first algorithm together with the algorithm in Section~\ref{sec:color-subprocedure} with a modification.

We map each point $p$ in $S$ to a plane $p^*$ in $\mathbb{R}^3$ using the standard duality transform,
and denote the set of all planes dual to the points in $S$ by $S^*$.
Due to the standard duality, our problem reduces to computing the convex hull of $\Lambda_\ell$, where
$\Lambda_\ell$ is the colorful level $\ell$ of the arrangement of $S^*$. 
We sort the planes as described in Section~\ref{sec:first-problem} and denote the sequence by
$\{h_1,\ldots,h_n\}$.
Let $H_j=\{h_i \cap h_j \mid 1 \leq i<j\}$.

We consider the planes in $S^*$ one by one in this order
and compute $K_j$ for each plane $h_j \in S^*$, where 
$\Gamma_j = \{K_i \cap h_j \mid 1\leq i < j\}$ and $K_j = \ch(\ch(\Gamma_j) \cup (\Lambda_\ell \cap h_j))$. 
To do this, we use a cutting of $H_j$ as described in Section~\ref{sec:first-problem}.
The difference is that we replace the procedure in Lemma~\ref{lem:procedure} with
the procedure described in \ref{sec:color-subprocedure} with a modification.

Finally, we obtain $K_j$ for each plane $h_j\in S^*$. By the construction of $K_j$, the convex hull
of $\Lambda_\ell$ is
the convex hull of the polygons $K_j$'s over all planes $h_j\in S^*$. Thus, we can obtain
the colorful level $\ell$ in $O(n^2\log^4 n)$ time in total and the colorful center region in the same time.

\begin{theorem}
	Given a set $S$ of $n$ colored points in $\mathbb{R}^3$ and an
	integer $\ell$, the set of points in $\mathbb{R}^3$ of colorful depth at most $\ell$
	with respect to $S$ can be computed in $O(n^2 \log^4 n)$ time.
\end{theorem}
\begin{corollary}
	Given a set $S$ of $n$ colored points in $\mathbb{R}^3$, the
	colorful center region of $S$ can be computed in $O(n^2 \log^4 n)$
	time.
\end{corollary}

\section{Conclusion}
In the first part of this paper, we presented an $O(n^2\log^4 n)$-time
algorithm for computing the center region in $\mathbb{R}^3$. This
algorithm is almost optimal since the combinatorial complexity of the
center region is $\Theta(n^2)$ in the worst case. Moreover, our
algorithm improves the previously best known algorithm which takes
$O(n^{2+\epsilon})$ time for any constant $\epsilon>0$.  In the second
part of this paper, we presented an $O(n\log^4 n)$-time algorithm for
computing the colorful center region in $\mathbb{R}^2$.
Both results were achieved by using the duality and the arrangement of lines
for the standard center region and the arrangement of convex polygonal chains
for the colorful center region.

As we mentioned in the introduction, it is not known whether
the center region for points in $d$-dimensional space can be computed efficiently
for $d>3$, except for an $O(n^d)$-time trivial algorithm 
which compute the arrangement of the hyperplanes dual to the points,
while a (colorful) Tukey median can be computed in $O(n\log n+n^{d-1})$
expected time. Since very large data of high dimensionality are common nowadays,
it is required to devise algorithms that compute center region for high dimensional
data efficiently.

{\bibliographystyle{abbrvnat} \bibliography{paper} }
\end{document}